\documentclass[draftcls,11pt,onecolumn,romanappendices]{IEEEtran}
\usepackage{times}
\usepackage{amsmath}  
\usepackage{amssymb} 
\usepackage{mathrsfs}
\usepackage{cite}
\usepackage{comment} 
\usepackage{upref}
\usepackage{amsfonts}
\usepackage{verbatim}
\usepackage[dvipsnames,usenames]{color}
\usepackage{enumerate}
\usepackage{graphicx}
\usepackage{subfigure}
\usepackage{latexsym}
\usepackage{bm}
\usepackage{multirow}
\usepackage{dsfont}
\usepackage{amsthm,xpatch}
\usepackage{mathtools}
\usepackage{bbm}
\usepackage{latexsym}
\usepackage{accents}
\usepackage{psfrag}
\usepackage{color}
\usepackage{times,color}
\usepackage[usenames,dvipsnames,svgnames,table]{xcolor}
\usepackage{pgf}
\usepackage{tikz}
\usetikzlibrary{arrows, automata, positioning, calc}
\usepackage{cancel} 
\usepackage{tikz, tikz-3dplot, amssymb}
\tdplotsetmaincoords{72}{130}

\makeatletter
\newif\if@borderstar
\def\bordermatrix{\@ifnextchar*{%
\@borderstartrue\@bordermatrix@i}{\@borderstarfalse\@bordermatrix@i*}%
}
\def\@bordermatrix@i*{\@ifnextchar[{\@bordermatrix@ii}{\@bordermatrix@ii[()]}}
\def\@bordermatrix@ii[#1]#2{%
\begingroup
\m@th\@tempdima8.75\p@\setbox\z@\vbox{%
\def\cr{\crcr\noalign{\kern 2\p@\global\let\cr\endline }}%
\ialign {$##$\hfil\kern 2\p@\kern\@tempdima & \thinspace %
\hfil $##$\hfil && \quad\hfil $##$\hfil\crcr\omit\strut %
\hfil\crcr\noalign{\kern -\baselineskip}#2\crcr\omit %
\strut\cr}}%
\setbox\tw@\vbox{\unvcopy\z@\global\setbox\@ne\lastbox}%
\setbox\tw@\hbox{\unhbox\@ne\unskip\global\setbox\@ne\lastbox}%
\setbox\tw@\hbox{%
$\kern\wd\@ne\kern -\@tempdima\left\@firstoftwo#1%
\if@borderstar\kern2pt\else\kern -\wd\@ne\fi%
\global\setbox\@ne\vbox{\box\@ne\if@borderstar\else\kern 2\p@\fi}%
\vcenter{\if@borderstar\else\kern -\ht\@ne\fi%
\unvbox\z@\kern-\if@borderstar2\fi\baselineskip}%
\if@borderstar\kern-2\@tempdima\kern2\p@\else\,\fi\right\@secondoftwo#1 $%
}\null \;\vbox{\kern\ht\@ne\box\tw@}%
\endgroup
}
\makeatother

\newtheorem{lemma}{Lemma}
\newtheorem{Corollary}{Corollary}
\newtheorem{theorem}{Theorem}

\newtheorem{proposition}{Proposition}
\newtheorem{definition}{Definition}

\usepackage{algorithm,algpseudocode}

\newcommand{\conv}{\operatorname{conv}}
\newcommand{\F}{\mathbb{F}}
\newcommand{\gauss}[3]{\genfrac{[}{]}{0pt}{}{#1}{#2}_{#3}}

\newcommand{\PG}{\operatorname{PG}}

\newcommand{\vek}[1]{\mathbf{#1}}

\newcommand{\hweight}{\mathrm{w}}

\newcommand{\card}[1]{\##1}

\makeatletter
\newcommand{\removelatexerror}{\let\@latex@error\@gobble}
\makeatletter
\newcommand{\proofpart}[2]{%
	\par
	\addvspace{\medskipamount}%
	\noindent\emph{Part #1: #2}\par\nobreak
	\addvspace{\smallskipamount}%
	\@afterheading
}
\makeatother

\makeatother

\makeatletter
\renewenvironment{proof}[1][\proofname]{\par
  \pushQED{\qed}%
  \normalfont \topsep6\p@\@plus6\p@\relax
  \trivlist
  \item[\hskip\labelsep
        \itshape
    #1\@addpunct{:}]\ignorespaces
}{%
  \popQED\endtrivlist\@endpefalse
}
\makeatother

\makeatletter

\renewcommand{\mathsf}[1]{#1}

\theoremstyle{definition}
\newtheorem{example}{Example}
\begin{document}

\pagestyle{plain}

\title{\LARGE \textbf{\Large A Geometric View of the Service Rates of Codes Problem and its \\[-0.75ex] Application to the Service Rate of the First Order Reed-Muller Codes}}

\author{\normalsize {$^\ast$}Fatemeh Kazemi, {$^\dagger$}Sascha Kurz, {$^\ddagger$}Emina Soljanin\\{\small {$^\ast$}Dept. of ECE, Texas A\&M University, USA (E-mail: fatemeh.kazemi@tamu.edu)}\\{\small {$^\dagger$} Dept. of Mathematics, University of Bayreuth, Germany (E-mail: sascha.kurz@uni-bayreuth.de)}\\ {\small {$^\ddagger$}Dept. of ECE, Rutgers University, USA (E-mail: emina.soljanin@rutgers.edu)}}

\maketitle 

\thispagestyle{plain}

\begin{abstract} 
Service rate is an important, recently introduced, performance metric associated with distributed coded storage systems. Among other interpretations, it measures the number of users that can be simultaneously served by the storage system. We introduce a geometric approach to address this problem. One of the most significant advantages of this approach over the existing approaches is that it allows one to derive bounds on the service rate of a code without explicitly knowing the list of all possible recovery sets. To illustrate the power of our geometric approach, we derive upper bounds on the service rates of the first order Reed-Muller codes and simplex codes. Then, we show how these upper bounds can be achieved. Furthermore, utilizing the proposed geometric technique, we show that given the service rate region of a code, a lower bound on the minimum distance of the code can be obtained.
\end{abstract}


\section{Introduction}

One of the most significant considerations in the design of distributed storage systems is serving a large number of users concurrently. The service rate has been recently recognized as an important performance metric that measures the number of users that can be simultaneously served by a storage system~\cite{noori2016storage,aktacs2017service,anderson2018service,peng2018distributed} that implements an erasure code. Maximizing the service rate reduces the latency experienced by users, particularly in a high traffic regime. 

The service rate problem considers a distributed storage system where $k$ files, ${f_1,\dots,f_k}$ are encoded into $n$, and stored across $n$ servers. File $f_i$ can be recovered by reading data from a single or a set of storage nodes, referred to as a recovery set for file $f_i$. Requests to download file $f_i$ arrive at rate $\lambda_i$, and can be split across its recovery sets. Server $l$ can simultaneously serve multiple requests if their cumulative arrival rate does not exceed the maximum service rate $\mu_l$ of server $l$. The service rate problem seeks to determine the service rate region of the coded storage system which is defined as the set of all request arrival rate vectors ${\boldsymbol{\lambda}=(\lambda_1,\dots,\lambda_k)}$ that can be served by this system.  

The service rate problem has been studied only in some special cases: 1) for MDS codes when $n \geq 2k$ and binary simplex codes in~\cite{aktacs2017service} and 2) for systems with arbitrary $n$ when ${k=2}$ in~\cite{aktacs2017service} and ${k=3}$ in~\cite{anderson2018service}. The existing techniques for solving the problem require enumeration of all possible recovery sets, which becomes increasingly complex when the number of files $k$ increases. Thus, introducing a technique which is not depending on the enumeration of recovery sets is of great significance. In this paper, we introduce a novel geometric approach with this goal in mind.

\subsection{Previous and Related Work}
The past two decades have seen an ever increasing interest in coding for storage and caching. Special codes that support efficient maintenance of storage under node failures have been proposed in e.g.,~\cite{dimakis2010network,dimakis2011survey,huang2013pyramid,gopalan2012locality,sardari2010memory}. The locality and availability of codes matter in such scenarios. This line of work mostly assumes infinite service rate (immediate service) for servers, and is primarily concerned with reliability of storage rather than with serving a large number of simultaneous users. 

Another line of work is focused on caching (see e.g.,~\cite{shanmugam2013femtocaching, maddah2016coding,hamidouche2014many}). In these work, the limited capacity of the backhaul link is considered as the main bottleneck, and the goal is usually to minimize its traffic by prefetching the popular contents at the storage nodes with limited size. Thus, these work do not address the scenarios such as live streaming, where many users wish to get the same content concurrently given the limited capacity of the access part of the network. 

The most related to this work are papers concerned with content download from coded storage. Load balancing in such systems has recently been addressed in~\cite{aktas2019load}. Memory allocation that maximizes the probability of successful content download under limited access to distributed storage was considered in e.g.,~\cite{allocation:sardariRFS10,allocation:LeongDH12} and references therein. Minimizing the service rate in these scenarios was considered in~\cite{noori2016storage, peng2018distributed}. This problem is similar to ours but it assumes that all users request the same content which is  encoded by an MDS code and stored on nodes with unlimited storage capacity. Fast content download from coded storage was considered in e.g.,~\cite{latency:JoshiSW15,latency:JoshiSW17}, and references therein. These papers strive to compute the download latency for increasingly complex queueing systems that appear in coded storage
\cite{Tompecs:AktasKS,ISIT:AktasS18,aktas2017simplex}. The service rate problem is related to the stability region of such queues.

\subsection{Main Contributions} 

We study the service rates of codes problem by introducing a novel geometric approach. This approach overcomes the main drawback of the previous work which are trying to solve this problem by formulating it as a sequence of linear programs (LP). There, one must exactly know all possible recovery sets to enumerate the constraints in each LP, and must also solve all the LPs.

Leveraging our novel geometric technique, we take initial steps towards deriving bounds on the service rates of some parametric classes of linear codes without explicitly listing the set of all possible recovery sets. In particular, we derive upper bounds on the service rates of the first order Reed-Muller codes and simplex codes, as two classes of codes which are most important in theory as well as in practice. Subsequently, we show how the derived upper bounds can be achieved. Moreover, utilizing the proposed geometric technique, we show that given the service rate region of a code, a lower bound on the minimum distance of the code can be derived.
\textit{The proofs of Lemmas and Corollaries can be found in the Appendix}.

\section{Problem Statement}\label{sec:Problem Formulation}

\subsection{Notation}
Throughout this paper, we denote vectors by bold-face small letters and matrices by bold-face capital letters. Let $\mathbb{N}$ denote the set of the non-negative integer numbers. Let $\mathbb{F}_q$ be a finite field for some prime power $q$, and $\mathbb{F}_q^n$ be the $n$-dimensional vector space over $\mathbb{F}_q$. Let us denote a q-ary linear code $\mathcal{C}$ of length $n$, dimension $k$ and minimum distance $d$ by $[n,k,d]_q$. We denote the Hamming weight of $\vek{x}$ by $\hweight(\vek{x})$. For a positive integer $k$, let $\vek{0}$ and $\vek{1}$ denote the all-zero and all-one column vectors of length $k$, respectively. Let $\vek{e}_i$ denotes a unit vector of length $k$, having a one at position $i$ and zeros elsewhere. For a positive integer $i$, define $[i]\triangleq\{1,\dots,i\}$. Let us denote the cardinality of a set or multiset $\mathcal{S}$ by $\#\mathcal{S}$.

\subsection{Service Rate of Codes}

Consider a storage system in which $k$ files $f_1,\dots,f_k$ are stored over $n$ servers, labeled $1,\dots,n$, using a linear $[n,k]_q$ code with generator matrix ${\mathbf{G} \in \mathbb{F}^{k\times n}_q}$. Let ${\vek{g}_j}$ denote the ${j}$th column of $\mathbf{G}$. A recovery set for the file $f_i$ is a set of stored symbols which can be used to recover file $f_i$. With respect to $\mathbf{G}$, a set ${R\subseteq [n]}$ is a recovery set for file $f_i$ if there exist ${\alpha_j}$'s ${\in \mathbb{F}_q}$ such that ${\sum_{j \in R}\alpha_j \mathbf{g}_j=\mathbf{e}_i}$, i.e., the unit vector $\mathbf{e}_i$ can be recovered by a linear combination of the columns of $\mathbf{G}$ indexed by the set $R$. W.l.o.g., we restrict our attention to the reduced recovery sets obtained by considering non-zero coefficients $\alpha_j$'s and linearly independent columns $\mathbf{g}_j$'s. 

Let $\mathcal{R}_{i}=\{R_{i,1},\dots,R_{i,t_i}\}$ be the $t_i\in \mathbb{N}$ recovery sets for file ${f_i}$. Let $\mu_l \in \mathbb{R}_{\geq 0}$ be the average rate at which the server ${l\in [n]}$ resolves received file requests. We denote the service rates of servers $1,\dots,n$ by a vector $\boldsymbol{\mu}=(\mu_1,\dots,\mu_n)$. We further assume that the requests to download file $f_i$ arrive at rate $\lambda_i$, $i\in [k]$. We denote the request rates for files $1,\dots,k$ by the vector $\boldsymbol{\lambda}=(\lambda_1,\dots,\lambda_k)$. We consider the class of scheduling strategies that assign a fraction of requests for a file to each of its recovery sets. Let $\lambda_{i,j}$ be the portion of requests for file $f_i$ that are assigned to the recovery set $R_{i,j}$, $j\in [t_i]$. 

The service rate region $\mathcal{S}(\mathbf{G},\boldsymbol{\mu}) \subseteq \mathbb{R}^k_{\geq 0}$ is defined as the set of all request vectors $\boldsymbol{\lambda}$ that can be served by a coded storage system with generator matrix $\mathbf{G}$ and service rate $\boldsymbol{\mu}$. Alternatively, $\mathcal{S}(\mathbf{G},\boldsymbol{\mu})$ can be defined as the set of all vectors $\boldsymbol{\lambda}$ for which there exist $\lambda_{i,j}\in \mathbb{R}_{\geq 0}$, $i\in [k]$ and $j\in [t_i]$, satisfying the following constraints:
\begin{subequations}\label{eq:1}
\begin{align}\label{condition_demand}
&\sum_{j=1}^{t_i} \lambda_{i,j}=\lambda_i, ~~~~~~~~~~~ \text{for all} ~~~ i\in [k], \\
&\sum_{i=1}^{k}~\sum_{\substack{{j\in [t_i]} \\ \label{condition_capacity} {l \in R_{i,j}}}} \lambda_{i,j} \leq \mu_l, ~~~ \text{for all} ~~~ l\in [n],\\
&\lambda_{i,j} \in \mathbb{R}_{\geq 0},~~~~~~~~~~~~~~ \text{for all} ~~~ i\in [k],~j\in [t_i].\label{eq:pos}
\end{align}
\end{subequations}
The constraints \eqref{condition_demand} guarantee that the demands for all files are served, and constraints \eqref{condition_capacity} ensure that no node receives requests at a rate in excess of its service rate.

\begin{lemma}\label{lem:convexity}
The service rate region $\mathcal{S}(\mathbf{G},\boldsymbol{\mu})$ is a non-empty, convex, closed, and bounded subset of $\mathbb{R}^k_{\geq 0}$.
\end{lemma}

\begin{proposition}\cite{rockafellar1970convex}\label{prop:convhull}
For any set ${\mathcal{A}=\{\mathbf{v}_1,\dots,\mathbf{v}_p\} \subseteq \mathbb{R}^k}$, the convex hull of the set $\mathcal{A}$, denoted by $\conv(\mathcal{A})$, consists of all convex combinations of the elements of $\mathcal{A}$, i.e., all vectors of the form $\sum_{i=1}^{p}\gamma_i\mathbf{v}_i$, with $\gamma_i \geq 0$, $\sum_{i=1}^{p}\gamma_i=1$.
\end{proposition}

\begin{Corollary}\label{lem:polytope}
The service rate region $\mathcal{S}(\mathbf{G},\boldsymbol{\mu}) \subseteq \mathbb{R}^k_{\geq 0}$ forms a polytope which can be expressed in two forms: as the intersection of a finite number of half spaces or as the convex hull of a finite set of vectors (vertices of the polytope).
\end{Corollary}

The service rate problem seeks to determine the service rate region $\mathcal{S}(\mathbf{G},\boldsymbol{\mu})$ of a coded storage system with generator matrix $\mathbf{G}$ and service rate $\boldsymbol{\mu}$. Based on Corollary~\ref{lem:polytope}, the first algorithm for computing the service rate region that comes to mind is enumerating all vertices of the polytope $\mathcal{S}(\mathbf{G},\boldsymbol{\mu})$ and then computing the convex hull of the resulting vertices. As we will show shortly, this problem can be formulated as an optimization problem consisting of a sequence of LPs.

Given that any $k-1$ request arrival rates, ${\lambda_{i_1},\dots,\lambda_{i_{k-1}}}$, are zeros, there exists a maximum value of ${\lambda_{i_k}}$, denoted by ${\lambda^\star_{i_k}}$, where ${0 \leq \lambda^\star_{i_k} \leq \sum_{l=1}^n \mu_l}$ (see the proof of Lemma~\ref{lem:convexity}) such that ${\lambda^\star_{i_k}.\mathbf{e}_{i_k} \in \mathcal{S}(\mathbf{G},\boldsymbol{\mu})}$ and all vectors ${\lambda_{i_k}.\mathbf{e}_{i_k}}$ with ${\lambda_{i_k} > \lambda^\star_{i_k}}$ are not in ${\mathcal{S}(\mathbf{G},\boldsymbol{\mu})}$. These constrained optimization problems of finding the maximum value ${\lambda^\star_{i_k}}$ are all LPs. For ${i \in [k]}$, let ${\mathbf{v}_i=\lambda^\star_{i}\vek{e}_i}$. Since ${\mathcal{J}=\{\mathbf{0},\mathbf{v}_1,\mathbf{v}_2,\dots,\mathbf{v}_k\} \subseteq \mathcal{S}(\mathbf{G},\boldsymbol{\mu})}$, as an immediate consequence of Lemma~\ref{lem:convexity} and Proposition~\ref{prop:convhull}, the set ${\conv(\mathcal{J})}$ is contained in ${\mathcal{S}(\mathbf{G},\boldsymbol{\mu})}$. Starting with ${\mathcal{J}}$, one can iteratively enlarge ${\mathcal{J}}$ until the subsequent procedure stops. A facet $H$ of ${\conv(\mathcal{J})}$ described by a vector $\mathbf{h}\in\mathbb{R}_{_\ge 0}^k$ and $\eta\in\mathbb{R}_{\ge 0}$ as ${H=\left\{\mathbf{x}\in\mathbb{R}^k_{\ge 0}\,:\, \mathbf{h}^\top \mathbf{x}=\eta\right\}\cap\conv(\mathcal{J})}$, is chosen.

For the chosen facet $H$ described by the vector $\mathbf{h}\in\mathbb{R}_{_\ge 0}^k$ and $\eta\in\mathbb{R}_{\ge 0}$, one should solve $\max \mathbf{h}^\top\boldsymbol{\lambda}$, where ${\boldsymbol{\lambda}\in\mathbb{R}_{\ge 0}^k}$ satisfies the demand constraints~(\ref{condition_demand}) and the capacity constraints~(\ref{condition_capacity}). If the optimal target value is strictly larger than $\eta$, then the solution vector $\boldsymbol{\lambda}^\star$ is added to $\mathcal{J}$ and this procedure continues. 
For any ${\mathbf{h}=(h_1,\dots,h_k)}$, the primal LP is given by
\begin{equation}\label{vertex_LP_primal}
\max~~{\sum_{i=1}^k  h_i\lambda_i} 
\qquad\text{s.t.}~~\eqref{eq:1}~~\text{holds}. 
\end{equation}\vspace{-0.35cm}
\newline
and the corresponding dual LP is given by
\begin{eqnarray}\label{vertex_LP_dual}
  &\min &{\sum_{l=1}^{n} \gamma_l\mu_l} \\
  &\text{s.t.} &h_i \leq \beta_i~~~~~~~~~~~~~~\forall i\in [k]\nonumber\\
  &&\beta_i \leq \sum_{l\in R_{i,j}} \gamma_l~~~~~~~~\forall i\in [k], \forall j\in [t_i]~ \nonumber\\
  &&\beta_i \in \mathbb{R}~~~~~~~~~~~~~~~~\forall i\in [k]\nonumber
  \\&&\gamma_l \in \mathbb{R}_{\geq 0}~~~~~~~~~~~~~~\forall l\in [n] \nonumber
\end{eqnarray}
 
Bsaed on the Duality Theorem~\cite{matousek2007understanding}, if both the primal LP and the corresponding dual LP have feasible solutions, then their optimal target values coincide. It can be easily confirmed that a feasible solution for the primal LP~\eqref{vertex_LP_primal} can be given by ${\lambda_{i,j}=0}$ and ${\lambda_i=0}$, and a feasible solution for the dual LP~\eqref{vertex_LP_dual} can be given by ${\beta_i=h_i}$ and $\gamma_l=\sum_{i=1}^k h_i$. Thus, given a generator matrix $\mathbf{G}$ of a linear code and a service rate $\boldsymbol{\mu}$ of the servers in a distributed coded storage system, the LP~\eqref{vertex_LP_primal} can be utilized to compute the maximum value of ${\eta=\sum_{i=1}^{k} h_i\lambda_i}$, denoted by ${\eta^\star}$, for every ${\mathbf{h}\in\mathbb{R}_{\geq 0}^k}$. Having $\eta^\star$ at hand, it is known that all ${\boldsymbol{\lambda}\in\mathcal{S}(\mathbf{G},\boldsymbol{\mu})}$ satisfy ${\sum_{i=1}^k h_i\lambda_i\le \eta^\star}$, which is a valid inequality for $\mathcal{S}(\mathbf{G},\boldsymbol{\mu})$. 

The downside of this approach is that we have to exactly know the set of all possible recovery sets for each file and also have to optimally solve all the LP~\eqref{vertex_LP_primal}. Using the dual LP~\eqref{vertex_LP_dual}, we run into a similar problem since in order to list all the inequalities in~\eqref{vertex_LP_dual}, again we require to know the elements of all the recovery sets for each file, which becomes increasingly complex when the number of files $k$ increases. Therefore, determining the service rate region of a code is a challenging problem, and in general we have to be pleased with lower and upper bounds. Thus, characterizing the exact service rate region of some parametric classes of linear codes or deriving some bounds on the service rate of a code without knowing explicitly all recovery sets is of great significance, which we aim to address in this paper. Towards this goal, we introduce a novel geometric approach. Leveraging our geometric approach, we derive upper bounds on the service rates of the first order Reed-Muller codes and simplex codes.

\subsection{Geometric View on Linear Codes~\cite{tsfasman1995geometric,dodunekov1998codes,beutelspacher1998projective}}\label{subsec:n-multiset}

\begin{definition}
For a vector space $\mathcal{V}$ of dimension $v$ over $\mathbb{F}_q$, ordered by inclusion, the set of all ${\mathbb{F}_q}$-subspaces of ${\mathcal{V}}$ forms a finite modular geometric lattice with meet ${X\wedge Y=X\cap Y}$, join $X\vee Y=X+Y$, and rank function $X\mapsto\dim(X)$. This subspace lattice of $\mathcal{V}$ is known as the projective geometry of ${\mathcal{V}}$, denoted by ${\PG(\mathcal{V})}$.
\end{definition}

For a vector space $\mathcal{V}$ of dimension $v$ over $\mathbb{F}_q$, the $1$-dimensional subspaces of $\mathcal{V}$ are the points of $\PG(\mathcal{V})$, the $2$-dimensional subspaces of $\mathcal{V}$ are the lines of ${\PG(\mathcal{V})}$, and the ${v-1}$ dimensional subspaces of $\mathcal{V}$ are called the hyperplanes of $\PG(\mathcal{V})$. The projective geometry ${\PG(\mathcal{V})}$ is also denoted by $\PG(v-1,q)$, which is referred to as the ${v-1}$ dimensional projective space over ${\mathbb{F}_q}$. This notion makes sense because of the fact that, up to isomorphism, the projective geometry ${\PG(\mathcal{V})}$ only depends on the order $q$ of the base field and the (\emph{algebraic}) dimension $v$, which is justifying the notion ${\PG(v-1,q)}$ of (\emph{geometric}) dimension $v-1$ over $\mathbb{F}_q$. 

Let $\mathcal{V}$ be a vector space of dimension $v$ over $\mathbb{F}_q$. The set of all $k$-dimensional subspaces of $\mathcal{V}$, referred to as \emph{$k$-subspaces}, will be denoted by $\gauss{\mathcal{V}}{k}{q}$. The cardinality of this set is given by the Gaussian binomial coefficient as follows

\[
\gauss{v}{k}{q} =
\begin{cases}
	\frac{(q^v-1)(q^{v-1}-1)\cdots(q^{v-k+1}-1)}{(q^k-1)(q^{k-1}-1)\cdots(q-1)} & \text{if }0\leq k\leq v\text{;}\\
	0 & \text{otherwise.}
\end{cases}
\]\vspace{0.1cm}

A multiset is a modification of the concept of a set that, unlike a set, allows for multiple instances for each of its elements. The positive integer number of instances, given for each element is called the multiplicity of this element in the multiset. More formally, a multiset $\mathcal{S}$ on a base set $\mathcal{X}$ can be identified with its characteristic function $\chi_{\mathcal{S}} : \mathcal{X} \to \mathbb{N}$, mapping $x \in \mathcal{X}$ to the multiplicity of $x$ in $\mathcal{S}$. The \emph{cardinality} of $\mathcal{S}$ is $\#\mathcal{S} = \sum_{x\in X} \chi_{\mathcal{S}}(x)$. $\mathcal{S}$ is also called \emph{$\#\mathcal{S}$-multiset}.

\begin{definition}
Let $\mathcal{V}$ be a vector space of dimension $v$ over $\mathbb{F}_q$, $\mathcal{P}$ be a multiset of points $p$ in $\PG(\mathcal{V})$ with characteristic function $\chi_{\mathcal{P}} : \PG(\mathcal{V}) \to \mathbb{N}$, and $\mathcal{H}$ denotes a hyperplane in $\PG(\mathcal{V})$. The restricted multiset $\mathcal{P} \cap \mathcal{H}$ is defined via its characteristic function as follows

\[
	\chi_{\mathcal{P} \cap \mathcal{H}}(p) = 
	\begin{cases}
	\chi_{\mathcal{P}}(p) & \text{if }p\in\gauss{\mathcal{H}}{1}{q}\text{;} \\
	0 & \text{otherwise.}
	\end{cases}
\] 
Then $\#(\mathcal{P} \cap \mathcal{H}) = \sum_{p\in\gauss{\mathcal{H}}{1}{q}} \chi_{\mathcal{P}}(p)$.
\end{definition}

Let $\mathbf{G} \in \mathbb{F}^{k\times n}_q$ be the generator matrix of a linear $[n,k]_q$ code $\mathcal{C}$, a $k$-subspace of the $n$-dimensional vector space $\mathbb{F}_q^n$. Let ${\mathbf{g}_i \in \mathbb{F}_q^k}$, ${i\in [n]}$ be the $i$th column of $\mathbf{G}$. Suppose that none of the ${\mathbf{g}_i}$'s is $\mathbf{0}$. (The code $\mathcal{C}$ is said to be of full length.) Then each $\mathbf{g}_i$ determines a point in the projective space ${\PG(k-1,q)}$, and $\mathcal{G}:=\{\mathbf{g}_1,\mathbf{g}_2,\dots,\mathbf{g}_n\}$ is a set of $n$ points in ${\PG(k-1,q)}$ if the $\mathbf{g}_i$ happen to be pair-wise independent. When dependence occurs, $\mathcal{G}$ is interpreted as a multiset and each point is counted with the appropriate multiplicity. In general, $\mathcal{G}$ is called $n$-multiset induced by $\mathcal{C}$.

\begin{proposition}
Different generator matrices of a code yield
projectively equivalent codes. In other words, there exist a bijective correspondence between the equivalence classes of full-length q-ary linear codes and the projective equivalence classes of multisets in finite projective spaces.
\end{proposition}

It should be noted that the importance of this correspondence lies in the fact that it relates the coding-theoretic properties of $\mathcal{C}$ to the geometric or the combinatorial properties of $\mathcal{G}$.

\begin{proposition}\label{Prop:mindis}
Let $\mathbf{G} \in \mathbb{F}^{k\times n}_q$ be the generator matrix of a linear $[n,k,d]_q$ code $\mathcal{C}$, and $\mathcal{G}$ be the $n$-multiset induced by code $\mathcal{C}$. The minimum distance $d$ of code $\mathcal{C}$ is given by
\begin{equation}
  \nonumber
  d=n-\max\card{(\mathcal{G}\cap\mathcal{H})},
\end{equation}
where $\mathcal{H}$ runs through all the hyperplanes of ${\PG(k-1,q)}$.
\end{proposition}

\begin{proof}
For an arbitrary non-zero row vector $\vek{a}=[a_1,\cdots,a_k]$ of dimension $k$, the Hamming weight of codeword $\vek{a}\mathbf{G}\in \mathcal{C}$ is given by 
\begin{equation}
  \nonumber
  \hweight(\vek{a}\mathbf{G})=n-\card{\{j\in
  [n];\vek{a}\vek{g}_j=0\}}
  =n-\card{(\mathcal{G}\cap\mathcal{A})},
\end{equation}
where $\mathcal{A}$ is a hyperplane in ${\PG(k-1,q)}$ with equation $a_1x_1+\dots+a_kx_k=0$. Thus, the codeword with minimum Hamming weight is resulted from a hyperplane $\mathcal{H}$ in ${\PG(k-1,q)}$ with maximum ${\card{(\mathcal{G}\cap\mathcal{H})}}$. The proof is completed considering the fact that the minimum distance of a code is equal to the minimum Hamming weight of its nonzero codewords. 
\end{proof}

\begin{example}
Consider the $k$-dimensional simplex code $\mathcal{C}$ over $\F_q$. In ${\PG(k-1,q)}$, the multiset $\mathcal{G}$ induced by code $\mathcal{C}$ has ${\gauss{k}{1}{q}}$ points, and all hyperplanes contain ${\gauss{k-1}{1}{q}}$ points. Thus, as an immediate consequence of Proposition~\ref{Prop:mindis} and its proof, every non-zero codeword of the corresponding linear code has a Hamming weight of $q^{k-1}$, which indicates that the minimum distance of code $\mathcal{C}$ is $q^{k-1}$. Let $\mathcal{H}$ be an arbitrary hyperplane in ${\PG(k-1,q)}$ and $\mathcal{P}$ be the set of all $q^{k-1}$ points of $\mathbb{F}_q^k$ that are not contained in $\mathcal{H}$. The corresponding code of $\mathcal{P}$ is known as a $k$-dimensional first order Reed-Muller code or as an affine $k$-dimensional simplex code. 
\end{example}

\subsection{First Order Reed-Muller Codes\cite{assmus1994designs,arikan2009channel,muller1954application,reed1953class}}\label{subsec:RM}
In this paper, we consider binary first order Reed-Muller codes $\text{RM}_2(1,k-1)$ with the integer parameter $k\ge 2$. It is known that $\text{RM}_2(1,k-1)$ is a linear $[2^{k-1},k,2^{k-2}]_2$ code. For a given $k$, one way of obtaining this code is to evaluate all multilinear polynomials with the binary coefficients, ${k-1}$ variables and the total degree of one on the elements of $\mathbb{F}_2^{k-1}$. The encoding polynomial for $\text{RM}_2(1,k-1)$ can be written as $c_1 +c_2\cdot Z_1+c_3\cdot Z_2+\dots +c_{k}\cdot Z_{k-1}$ where $Z_1,\dots,Z_{k-1}$ are the $k-1$ variables, and $c_1,\dots,c_k$ are the binary coefficients of this polynomial. Indeed, the data symbols $f_1,\dots,f_k$ are used as the coefficients of the encoding polynomial, and the codeword symbols are obtained by evaluating the encoding polynomial on all vectors ${(Z_1,\dots,Z_{k-1})\in \mathbb{F}_2^{k-1}}$.

Another way of describing $k$-dimensional binary first order Reed-Muller codes $\text{RM}_2(1,k-1)$ is based on the generator matrix which can be constructed as follows. Let write the set of all $(k-1)$-dimensional binary vectors as $\mathcal{X}=\mathbb{F}_2^{k-1}=\{\vek{x}_1,\dots,\vek{x}_n\}$ where $n=2^{k-1}$ and for $i \in [n]$, $\vek{x}_i=(x_{i_{k-1}},\dots,x_{i_{1}})$ with $x_{i_j} \in \F_2$, $j \in [k-1]$. For any ${\mathcal{A} \subseteq \mathcal{X}}$, define the indicator vector $\mathbb{I}_\mathcal{A} \in \mathbb{F}_2^{k-1}$ as follows,
\[
	(\mathbb{I}_\mathcal{A})_i=
	\begin{cases}
	1 & \text{if }\vek{x}_i\in \mathcal{A}\text{;} \\
	0 & \text{otherwise.}
	\end{cases}
\]\vspace{0.05cm}

For the $k$ rows of the generator matrix of $\text{RM}_2(1,k-1)$, define $k$ row vectors of length $2^{k-1}$ as follows, $\vek{r}_0=(1,\dots,1)$ and ${\vek{r}_j=\mathbb{I}_{\mathcal{H}_j}}$, ${j \in [k-1]}$, where ${\mathcal{H}_j=\{\vek{x}_i \in \mathcal{X} \mid x_{i_j}=0\}}$. The set $\{\vek{r}_{k-1},\dots,\vek{r}_1,\vek{r}_0\}$ defines the rows of a non-systematic generator matrix of the $\text{RM}_2(1,{k-1})$. Note that for a systematic generator matrix of the $\text{RM}_2(1,{k-1})$, the set of rows $\{\vek{r}_{k-1},\dots,\vek{r}_1,\sum_{i=0}^{k-1}\vek{r}_i\}$ can be considered.  
\begin{example}
Consider $\text{RM}_2(1,3)$ which is an $[8,4,4]_2$ code.  We first define the set $\mathcal{X}$ as follows
\[{\mathcal{X}=\mathbb{F}_2^{3}=\{(0,0,0),(0,0,1),\dots,(1,1,1)\}}=\{\vek{x}_1,\dots,\vek{x}_8\}\]
It then follows that ${\mathcal{H}_3=\{\vek{x}_1,\vek{x}_2,\vek{x}_3,\vek{x}_4\}}$ that gives ${\vek{r}_3=(1,1,1,1,0,0,0,0)}$, and $\mathcal{H}_2=\{\vek{x}_1,\vek{x}_2,\vek{x}_5,\vek{x}_6\}$ which gives $\vek{r}_2=(1,1,0,0,1,1,0,0)$, and ${\mathcal{H}_1=\{\vek{x}_1,\vek{x}_3,\vek{x}_5,\vek{x}_7\}}$ which results ${\vek{r}_1=(1,0,1,0,1,0,1,0)}$. Let ${\vek{r}_0}$ be the all-one row vector of dimension eight. The set ${\{\vek{r}_{3},\vek{r}_{2},\vek{r}_1,\vek{r}_0\}}$ defines the rows of a non-systematic generator matrix of the ${\text{RM}_2(1,3)}$ as follows 

\[
     \mathbf{G}= \begin{bmatrix}
      1 & 1 & 1 & 1 & 0 & 0 & 0 & 0 \cr
      1 & 1 & 0 & 0 & 1 & 1 & 0 & 0 \cr
      1 & 0 & 1 & 0 & 1 & 0 & 1 & 0 \cr
      1 & 1 & 1 & 1 & 1 & 1 & 1 & 1 \cr
    \end{bmatrix}
   \]\
   
Also, ${\sum_{i=0}^{3}\vek{r}_i=(0,1,1,0,1,0,0,1)}$, and a systematic generator matrix of the $\text{RM}_2(1,3)$ is given by

\[
     \mathbf{G}= \begin{bmatrix}
      1 & 1 & 1 & 1 & 0 & 0 & 0 & 0 \cr
      1 & 1 & 0 & 0 & 1 & 1 & 0 & 0 \cr
      1 & 0 & 1 & 0 & 1 & 0 & 1 & 0 \cr
      0 & 1 & 1 & 0 & 1 & 0 & 0 & 1 \cr
    \end{bmatrix}
   \]\
\end{example}

\section{Geometric View on Service Rate of Codes}

In this section, we use the geometric description of linear codes. For a linear code $\mathcal{C}$ with generator matrix ${\mathbf{G} \in \mathbb{F}^{k\times n}_q}$, we consider the $n$-multiset $\mathcal{G}$ induced by $\mathcal{C}$ in ${\PG(k-1,q)}$ with the characteristic function $\chi_{\mathcal{G}}$ as defined in the section~\ref{subsec:n-multiset}. Thus, each point ${p\in\PG(k-1,q)}$ has a certain multiplicity ${\chi_{\mathcal{G}}(p)\in\mathbb{N}}$. In this language, the reduced recovery sets are subsets of $\mathcal{G}$ such that each point can be taken once in a reduced recovery set. Also, the service rate of each point ${p}$, denoted by $\mu(p)$, can be defined as the sum of the service rates of the nodes (columns of $\mathbf{G}$) corresponding to the point $p$. Based on this definition, $\mu(p)=\sum_{l\in \mathcal{L}_p}\mu_l$ where $\mathcal{L}_p$ is the set of nodes that correspond to the same point $p \in \PG(k-1,q)$. Since $\#\mathcal{L}_p=\chi_{\mathcal{G}}(p)$, if all nodes in the set $\mathcal{L}_p$ have the same service rate, say $\mu_p$, then we have $\mu(p)=\chi_{\mathcal{G}}(p)\cdot\mu_p$. 

\begin{lemma}\label{lemma_hyperplane_constraint}
Let ${\mathbf{G} \in \mathbb{F}^{k\times n}_q}$ be the generator matrix 
of an $[n,k]_q$ code $\mathcal{C}$, and $\mathcal{G}$ be the $n$-multiset induced by code $\mathcal{C}$ with service rate $\mu(p)$ of each point ${p\in\PG(k-1,q)}$. If for some ${i \in [k]}$, ${s\cdot \vek{e}_i\in\mathcal{S}(\mathbf{G},\boldsymbol{\mu})}$ and a hyperplane $\mathcal{H}$ of $\PG(k-1,q)$ is not containing $\vek{e}_i$, then we have
\[
  s \leq \sum_{p \in {\PG(k-1,q) \setminus \mathcal{H}}} \mu(p).
\]   
\end{lemma}

\begin{Corollary}\label{cor_min_distance}
Let $\mathbf{G} \in \mathbb{F}^{k\times n}_q$ be the generator matrix of a linear $[n,k,d]_q$ code $\mathcal{C}$ with service rate ${\mu_l=1}$ of all nodes $l \in [n]$, and $\mathcal{G}$ be the $n$-multiset induced by code $\mathcal{C}$. If for all $i \in [k]$, $s\cdot \vek{e}_i\in\mathcal{S}(\mathbf{G},\boldsymbol{\mu})$, then the minimum distance $d$ of code $\mathcal{C}$ is at least $\lceil s\rceil$. 
\end{Corollary}


\begin{Corollary}\label{cor_hyperplane_constraint}
Let ${\mathbf{G} \in \mathbb{F}^{k\times n}_q}$ be the generator matrix of a linear $[n,k]_q$ code $\mathcal{C}$, and $\mathcal{G}$ be the $n$-multiset induced by code $\mathcal{C}$ with service rate $\mu(p)$ of each point ${p\in\PG(k-1,q)}$. Let $\mathcal{I}\subseteq [k]$. If for all $i\in\mathcal{I}$, there exist ${s_i\in\mathbb{R}_{\ge 0}}$ such that ${\sum_{i\in\mathcal{I}} s_i\cdot \vek{e}_i\in\mathcal{S}(\mathbf{G},\boldsymbol{\mu})}$ and a hyperplane $\mathcal{H}$ of $\PG(k-1,q)$ which is not containing $\vek{e}_i$ for all $i\in\mathcal{I}$, then we have 
\[
  s \leq \sum_{p \in {\PG(k-1,q) \setminus \mathcal{H}}} \mu(p).
\] where $s=\sum_{i\in\mathcal{I}} s_i$.
\end{Corollary}

It should be noted that Corollary~\ref{cor_hyperplane_constraint} enables us to derive upper bounds on the service rate of the first order Reed-Muller codes and simplex codes. In what follows, without loss of generality, we assume that the service rate of all servers in the coded storage system is $1$, i.e., $\mu_l=1$ for all $l \in [n]$. Thus, by this assumption, the service rate region of a code only depends on the generator matrix $\mathbf{G}$ of the code and can be denoted by $\mathcal{S}(\mathbf{G})$.

\section{Service Rate Region of Simplex Codes}

In this section, by leveraging a novel geometric approach, we characterize the service rate region of the binary simplex codes which are special rate-optimal subclass of availability codes that are known as an important family of distributed storage codes. As we will show, the determined service rate region coincides with the region derived in~\cite[Theorem 1]{aktacs2017service}.

\begin{theorem}\label{thm_service_region_simplex_code}
For each integer ${k\ge 1}$, the service rate region of the $k$-dimensional binary simplex code $\mathcal{C}$, which is a linear ${[2^k-1,k,2^{k-1}]_2}$ code with generator matrix $\mathbf{G}$ is given by \[\mathcal{S}(\mathbf{G})=\left\{\boldsymbol{\lambda}\in\mathbb{R}^k_{\ge 0}\,:\, \sum_{i=1}^k \lambda_i\le 2^{k-1}\right\}.\]
\end{theorem}
\begin{proof}
Note that the simplex code is projective. Since the projective space $\PG(k-1,2)$ contains exactly $2^k-1$ points, the generator matrix $\mathbf{G}$ consists of all non-zero vectors of $\mathbb{F}_2^k$. (Up to column permutations the generator matrix is unique.) Given an arbitrary $i \in [k]$, we partition the columns of $\mathbf{G}$ into $\vek{e}_i$ and $\left\{\vek{x},\vek{x}+\vek{e}_i\right\}$ for all $2^{k-1}-1$ non-zero vectors $\vek{x}\in\mathbb{F}_2^k$ with $i$th coordinate being equal to zero. Thus, for all ${i \in [k]}$, ${2^{k-1}\cdot \vek{e}_i\in\mathcal{S}(\mathbf{G})}$. Let $\mathbf{v}_i={2^{k-1}\cdot \vek{e}_i}$ for $i \in [k]$. Since ${\mathcal{J}=\{\mathbf{0},\mathbf{v}_1,\mathbf{v}_2,\dots,\mathbf{v}_k\} \subseteq \mathcal{S}(\mathbf{G})}$, based on Lemma~\ref{lem:convexity} and Proposition~\ref{prop:convhull}, the $\conv(\mathcal{J})$ is contained in $\mathcal{S}(\mathbf{G})$, i.e.,
\[\mathcal{S}(\mathbf{G})\supseteq \left\{\boldsymbol{\lambda}\in\mathbb{R}^k_{\ge 0}\,:\, \sum_{i=1}^k \lambda_i\le 2^{k-1}\right\}\] 
For the other direction, we consider the hyperplane $\mathcal{H}$ given by $\sum_{i=1}^k x_i=0$, which does not contain any unit vector $\vek{e}_i$. Thus, for any demand vector $\boldsymbol{\lambda}=(\lambda_1,\dots,\lambda_k)$ in the service rate region, the  Corollary~\ref{cor_hyperplane_constraint} results in $\sum_{i=1}^k \lambda_i\le 2^{k-1}$. The reason is that half of the vectors in $\mathbb{F}_2^k$ which are the columns of $\mathbf{G}$ and so the elements of  $\mathcal{G}$, are not contained in  $\mathcal{H}$.
\end{proof}

\section{Service Rate Region of Reed-Muller Codes}
This section seeks to characterize the service rate region of the $\text{RM}_2(1, k-1)$ code with a non-systematic and a systematic generator matrix $\mathbf{G}$ constructed as described in section~\ref{subsec:RM}. 

\subsection{Non-Systematic First Order Reed-Muller Codes}

\begin{theorem}
For each integer $k\ge 2$, the service rate region of first order Reed-Muller code ${\text{RM}_2(1, k-1)}$ (or binary affine $k$-dimensional simplex code) with a non-systematic generator matrix $\mathbf{G}$ constructed as described in section~\ref{subsec:RM}, if ${k\in \{2,3\}}$ is given by
\begin{align*}
  S(\mathbf{G}) &= \left\{ \lambda \in \mathbb{R}^k_{\geq 0} \,:\, \sum_{i=1}^{k}\lambda_i \leq 2^{k-2} \right\}=\operatorname{conv}\left(\{\mathbf{0},\mathbf{v}_1,\dots,\mathbf{v}_{k}\} \right)
\end{align*}
and if ${k\ge 4}$, is given by
\begin{align*}
S(\mathbf{G})&= \Bigl\{ \lambda \in \mathbb{R}^k_{\geq 0} \,:\, \sum_{i=1}^{k}\lambda_i \leq 2^{k-2},
   \sum_{i=1}^{k-1}\lambda_i + \frac{3}{2} \lambda_k -1 \leq 2^{k-2}\Bigr\}\\
   &= \operatorname{conv}\left(\{\mathbf{0},\mathbf{v}_1,\dots,\mathbf{v}_{k-1},\mathbf{u}_k,\mathbf{w}_1,\dots,\mathbf{w}_{k-1}\}\right),
\end{align*}
where ${\mathbf{v}_i=2^{k-2}\cdot \vek{e}_i}$ for $i \in [k]$ and ${\mathbf{w}_j={(2^{k-2}-2)\cdot \vek{e}_j+2\cdot \vek{e}_k}}$ for $j \in [k-1]$. Also, $\mathbf{u}_{k}=\tfrac{2^{k-1}+2}{3}\cdot \vek{e}_k$.
\end{theorem}

\begin{proof}
The proof consists of a converse and an achievability. 

\underline{Converse}: The unit vector $\vek{e}_i$ for all ${i \in [k-1]}$ is not a column of $\mathbf{G}$ which means that file $f_i$ does not have any systematic recovery set. Therefore, for file $f_i$, ${i\in [k-1]}$, all recovery sets have cardinality at least two, and the minimum system capacity utilized by ${\lambda_i}$, ${i \in [k-1]}$, is ${2\lambda_i}$. For file $f_k$, the cardinality of every reduced recovery set is odd since all columns of generator matrix $\mathbf{G}$ has one in the last row. Hence, for file $f_k$, the unit vector $\vek{e}_k$ that is a column of $\mathbf{G}$, forms a systematic recovery set of cardinality one, while all other recovery sets have cardinality at least three. Hence, the minimum capacity used by ${\lambda_k \geq 1}$ is ${1+3(\lambda_k-1)}$. Since the system has $2^{k-1}$ servers, each of service rate (capacity) $1$, based on the capacity constraints, the total capacity utilized by the requests for download must be less than $2^{k-1}$. Thus, any vector $\boldsymbol{\lambda}=(\lambda_1,\dots,\lambda_k)$ in the service rate region must satisfy the following valid constraint,
\begin{equation}
  \label{ie_upper_1}
  \sum_{i=1}^{k-1}\lambda_i + \frac{3}{2} \lambda_k -1 \leq 2^{k-2}
\end{equation}
Consider the hyperplane $\mathcal{H}$ given by $\sum_{i=1}^k x_i=0$, that does not contain any unit vector $\vek{e}_i$. The columns of generator matrix $\mathbf{G}$ and so the elements of $\mathcal{G}$ which are not contained in $\mathcal{H}$, are the vectors in $\mathbb{F}_2^k$ with one in the last coordinate that satisfy $\sum_{i=1}^{k-1} x_i=0$. It is easy to see that there are $2^{k-2}$ such vectors. Thus, applying Corollary~\ref{cor_hyperplane_constraint} for hyperplane $\mathcal{H}$ impose another valid constraint as follows that any demand vector $\boldsymbol{\lambda}=(\lambda_1,\dots,\lambda_k)$ in the service rate region must satisfy, 
\begin{equation}
  \label{ie_upper_2}
  \sum_{i=1}^{k}\lambda_i \leq 2^{k-2}
\end{equation}
It should be noted that for ${\lambda_k<2}$, the Inequality~(\ref{ie_upper_2}) is tighter than~(\ref{ie_upper_1}), while for ${\lambda_k>2}$ Inequality~(\ref{ie_upper_1}) is tighter than~(\ref{ie_upper_2}). This means that for $k\in \{2,3\}$ Inequality~(\ref{ie_upper_1}) is redundant.

\underline{Achievability}: For the other direction, we will provide solutions (constructions) for the vertices of the corresponding polytope as follows. Let ${\mathcal{R}'\subseteq\mathbb{F}_2^k}$, ${|\mathcal{R}'|=2^{k-1}}$ be the set of columns of $\mathbf{G}$ with one in the last coordinate. For all ${i \in [k-1]}$, consider all the $2^{k-2}$ vectors ${\vek{x}\in \mathcal{R}'}$ with zero in the $i$th coordinate, then ${\vek{x}+\vek{e}_i\in\mathcal{R}'}$, and so ${\{\vek{x},\vek{x}+\vek{e}_i\}}$ constitutes a recovery set of cardinality two for file $f_i$. Thus, for each file $f_i$, ${i \in [k-1]}$, the columns of $\mathbf{G}$ can be partitioned into ${2^{k-2}}$ pairs ${\{\vek{x},\vek{x}+\vek{e}_i\}}$ which determines ${2^{k-2}}$ disjoint recovery sets for file $f_i$, ${i \in [k-1]}$. Therefore, the demand vectors ${2^{k-2}\cdot \vek{e}_i}$ for all $i \in [k-1]$ can be satisfied, i.e., ${2^{k-2}\cdot \vek{e}_i \in S(\mathbf{G})}$. For file $f_k$, there are exactly one systematic recovery set of cardinality one which is the column $\vek{e}_k$ of $\mathbf{G}$, and ${(2^{k-1}-1).(2^{k-1}-2)/6}$ recovery sets of cardinality three which are the sets ${\{\vek{x},\vek{x'},\vek{x}+\vek{x'}+\vek{e}_k\}}$ for all pairs ${{\vek{x}, \vek{x'} \in \mathcal{R}' \setminus \vek{e}_k}}$. Note that for ${k=2}$, according to Inequality~(\ref{ie_upper_2}), one can readily confirm that ${\lambda_k \leq 1}$. Thus, for ${k=2}$ the systematic recovery set of file $f_k$ can be utilized for satisfying the demand vector $1\cdot \vek{e}_k$. For ${k\ge 3}$, it should be noted that that each column ${\vek{x} \in \mathcal{R}' \setminus \vek{e}_k}$ is contained in exactly ${(2^{k-1}-2)/2}$ recovery sets of file ${f_k}$ of cardinality three. Since the capacity of each node is one, from each recovery set the request rate of ${1/(2^{k-2}-1)}$ can be satisfied without violating the capacity constraints. Thus, the demand vector ${\tfrac{2^{k-1}+2}{3}\cdot \vek{e}_k}$ can be satisfied. For the remaining part, we consider $k\ge 4$. Let ${i,j\in [k-1]}$ with ${i\neq j}$ be arbitrary. With this ${\{\vek{e}_k,\vek{e}_i+\vek{e}_k\}}$ and ${\{\vek{e}_j+\vek{e}_k,\vek{e}_i+\vek{e}_j+\vek{e}_k\}}$ are two of ${2^{k-2}}$ recovery sets of cardinality two for file ${f_i}$. Thus, the elements in ${\mathcal{R}'\backslash\left\{\vek{e}_k,\vek{e}_i+\vek{e}_k,\vek{e}_j+\vek{e}_k,\vek{e}_i+\vek{e}_j+\vek{e}_k\right\}}$ can be partitioned into ${2^{k-2}-2}$ recovery sets for file ${f_i, i\in [k-1]}$. Also, the sets $\{\vek{e}_k\}$ and ${\left\{\vek{e}_i+\vek{e}_k,\vek{e}_j+\vek{e}_k,\vek{e}_i+\vek{e}_j+\vek{e}_k\right\}}$ can be utilized as two disjoint recovery sets for file $f_k$. Thus, the demand vector ${\left(2^{k-2}-2\right)\cdot \vek{e}_i+2\cdot \vek{e}_k}$ can be satisfied.
\end{proof}

\subsection{Systematic First Order Reed-Muller Codes}

\begin{theorem}\label{theorem:systematic-RM}
For each integer ${k\ge 2}$, the service rate region of first order Reed-Muller code ${\text{RM}_2(1, k-1)}$ (or binary affine $k$-dimensional simplex code) with a systematic generator matrix $\mathbf{G}$ constructed as described in section~\ref{subsec:RM}, if ${k=2}$ is given by
\[
  \mathcal{S}(\mathbf{G})=\left\{ \lambda \in \mathbb{R}^k_{\geq 0} \,:\, \lambda_1 \leq 1, \lambda_2\leq 1 \right\}
  =\operatorname{conv}\left(\mathbf{0}, \vek{e}_1+\vek{e}_2 \right)
\]
if ${k=3}$, is given by
\begin{align*}
  \mathcal{S}(\mathbf{G})= \Big\{ \lambda \in \mathbb{R}^k_{\geq 0} \,:\, -\lambda_i+\sum_{j=1}^3 \lambda_j\le 2,\forall i \in [k]\Big\} =\operatorname{conv}\left(\mathbf{0}, 2\cdot \vek{e}_1, 2\cdot \vek{e}_2, 2\cdot \vek{e}_3, \vek{e}_1+\vek{e}_2+\vek{e}_3 \right)
\end{align*}
if ${k=4}$, is given by
\begin{align*}
  \mathcal{S}(\mathbf{G})&={\Big\{{\lambda \in \mathbb{R}^k_{\geq 0} \,:\,} {-\lambda_i+\sum_{j=1}^k \lambda_j\le 4}, {2\lambda_i+\sum_{j=1}^k\lambda_j \le 10}\, {\forall i \in [k]} \Big\}} \\
  &=\operatorname{conv}\left(\mathbf{0}, \mathbf{p}_i\,\forall i\in [k], \mathbf{q}_{i,j}\,\forall i,j\in [k]\text{ with }i\neq j, \tfrac{4}{3}\cdot\mathbf{1} \right)
\end{align*}
and if ${k\geq 5}$, $\mathcal{S}(\mathbf{G})$ lies inside the region given by
\begin{align*}
  {\mathcal{S}(\mathbf{G}) \subseteq} \, {\Big\{ {\lambda \in \mathbb{R}^k_{\geq 0} \,:\,} {\sum_{i\in[k]\setminus \mathcal{S}} \lambda_i} +{\sum_{j\in \mathcal{S}} (3\lambda_j-2)} \,\le\, {2^{k-1}} \, {\forall \mathcal{S}\subseteq [k]}\Big\}}. 
\end{align*}
where ${\mathbf{p}_i=\tfrac{10}{3}\cdot \vek{e}_i}$ and ${\mathbf{q}_{i,j}={3\cdot \vek{e}_i+1 \cdot \vek{e}_j}}$ for ${i,j \in [k]}$.
\end{theorem}

\begin{proof}
Based on the construction described in section~\ref{subsec:RM} for a systematic generator matrix $\mathbf{G}$ of the ${\text{RM}_2(1, k-1)}$, it can be seen that the number of ones in each column of $\mathbf{G}$ is odd, and the constructed systematic generator matrix, up to column permutations, is unique. Let the columns of $\mathbf{G}$ which are the set of vectors in ${\mathbb{F}_2^k}$ with odd number of ones, be denoted by ${\mathcal{R}'\subseteq\mathbb{F}_2^k}$, ${|\mathcal{R}'|=2^{k-1}}$.  

\underline{Converse}: For an arbitrary file $f_i$, $i{\in [k]}$, the unit vector $\vek{e}_i$ is a column of $\mathbf{G}$ that forms a systematic recovery set of cardinality one, while all other recovery sets have cardinality at least three. The proof is based on the contradiction approach. Let ${\vek{x},\vek{x}' \in \mathcal{R}'\setminus \vek{e}_i}$. Assume that $\{\vek{x},\vek{x}'\}$ forms a recovery set of cardinality two for file $f_i$, i.e., ${\vek{x}+\vek{x}'=\vek{e}_i}$. Since both $\vek{x}$ and $\vek{x}'$ have an odd number of ones, their sum must have an even number of ones which is a contradiction. Indeed, for all pairs ${\vek{x},\vek{x}' \in \mathcal{R}'\setminus \vek{e}_i}$, the set $\left\{\vek{x},\vek{x}',\vek{x}+\vek{x}'+\vek{e}_i\right\}$ forms a recovery set of cardinality three for file $f_i$, ${i\in [k]}$. Thus, if ${\lambda_i \leq 1}$, the requests for file $f_i$ can be fully satisfied by the systematic recovery set $\{\vek{e}_i\}$ and the system capacity utilized by ${\lambda_i}$ is ${\lambda_i}$. However, for ${\lambda_i \geq 1}$, the system capacity utilized by $\lambda_i$ is at least ${1+3(\lambda_i-1)=3\lambda_i-2}$. Since the system has $2^{k-1}$ servers of capacity $1$, the following constraints are valid constraints such that any vector $\boldsymbol{\lambda}=(\lambda_1,\dots,\lambda_k)$ in the service rate region must satisfy: 
\begin{equation}
  \label{ie_upper_bound_1_rm_systematic}
  \sum_{i\in[k]\backslash \mathcal{S}} \lambda_i +\sum_{j\in \mathcal{S}} (3\lambda_j-2) \,\le\, 2^{k-1}~~~\forall \mathcal{S} \subseteq [k]
\end{equation}

Applying Corollary~\ref{cor_hyperplane_constraint} on all hyperplanes ${\mathcal{H}_j}$, ${j \in [k]}$, given by ${\sum_{i\in [k]\setminus j}x_i=0}$, where each hyperplane ${\mathcal{H}_j}$, ${j \in [k]}$ does not contain any unit vectors $\vek{e}_i$, ${i \in [k]\setminus j}$, yields another set of valid constraints on any demand vector $\boldsymbol{\lambda}=(\lambda_1,\dots,\lambda_k)$ in the service rate region as follows: 
\begin{equation}
  \label{ie_upper_bound_2_rm_systematic}
  \sum_{i\in [k]\setminus j}\lambda_i \leq 2^{k-2} ~~~~\forall j \in [k]
\end{equation}

Note that for ${k\in \{2,3\}}$, Inequality~(\ref{ie_upper_bound_2_rm_systematic}) is tighter than~(\ref{ie_upper_bound_1_rm_systematic}). For ${k=2}$, Inequality~(\ref{ie_upper_bound_2_rm_systematic}) gives ${\lambda_1\le 1}$ and ${\lambda_2\le 1}$. For ${k=3}$, Inequality~(\ref{ie_upper_bound_2_rm_systematic}) gives ${\sum_{i=1}^{3}\lambda_i-\lambda_i\le 2}$ for all ${i\in [3]}$. Summing up these three inequalities and dividing them by two, results ${\sum_{i=1}^{3}\lambda_i\le 3}$. For ${k=4}$, Inequality~(\ref{ie_upper_bound_2_rm_systematic}) yields ${\sum_{i=1}^{4}\lambda_i-\lambda_i\le 4}$ for all ${i\in [4]}$. Summing up these four inequalities and dividing by three gives ${\sum_{i=1}^{4}\lambda_i\le \tfrac{16}{3}}$. Also,  for ${k=4}$, Inequality~(\ref{ie_upper_bound_1_rm_systematic}) gives a set of constraints, among which the constraints ${\sum_{i=1}^{4}\lambda_i+2\cdot \lambda_i\le 10}$ for all ${i\in [4]}$, are tighter than the ones already obtained from~(\ref{ie_upper_bound_2_rm_systematic}) in some region. For ${k\geq 5}$, Inequality~(\ref{ie_upper_bound_1_rm_systematic}) is always tighter than~(\ref{ie_upper_bound_2_rm_systematic}).
  
\underline{Achievability}: For ${k \leq 4}$, we have to provide constructions for the vertices of the corresponding polytope. As discussed, for each file $f_i$, with ${i \in [k]}$, there are exactly one systematic recovery set of cardinality one which is the column $\vek{e}_i$ of $\mathbf{G}$, and ${(2^{k-1}-1).(2^{k-1}-2)/6}$ recovery sets of cardinality three which are the sets of the form ${\{\vek{x},\vek{x'},\vek{x}+\vek{x'}+\vek{e}_i\}}$ for all pairs ${{\vek{x}, \vek{x'} \in \mathcal{R}' \setminus \vek{e}_i}}$. For $k=2$, the two disjoint recovery sets $\{\vek{e}_1\}$ and $\{\vek{e}_2\}$, which are the only recovery sets for files $f_1$ and $f_2$, respectively, can be used to satisfy the demand vector ${\vek{e}_1+\vek{e}_2}$. Now, consider $k\ge 3$. Since each column ${\vek{x} \in \mathcal{R}' \setminus \vek{e}_i}$ is contained in exactly ${(2^{k-1}-2)/2}$ recovery sets of file ${f_i}$, ${i \in [k]}$ of cardinality three, and the capacity of each node is one, from each recovery set the request rate of ${1/(2^{k-2}-1)}$ can be satisfied without violating the capacity constraints. Thus, the demand vector ${\tfrac{2^{k-1}+2}{3}\cdot \vek{e}_i}$ for all $i \in [k]$ can be satisfied. This means that for $k=3$ and $k=4$, respectively the demand vectors ${2\cdot \vek{e}_i}$ for all $i \in [3]$, and ${\tfrac{10}{3}\cdot \vek{e}_i}$ for all $i \in [4]$ can be satisfied. Also, for $k=3$, the demand vector $\vek{e}_1+\vek{e}_2+\vek{e}_3$ can be achieved by the disjoint systematic recovery sets $\{\vek{e}_1\}$, $\{\vek{e}_2\}$, and $\{\vek{e}_3\}$. Now, let assume $k\geq 4$. Let ${i,j\in [k]}$ with ${i\neq j}$ be arbitrary. The systematic recovery sets $\{\vek{e}_i\}$ and $\{\vek{e}_j\}$ can be used for files ${f_i}$ and ${f_j}$, respectively. Additionally, consider all the ${(2^{k-2}-1).(2^{k-1}-4)/3}$ recovery sets ${\{\vek{x},\vek{x'},\vek{x}+\vek{x'}+\vek{e}_i\}}$ of cardinality three for file $f_i$ that do not contain $\vek{e}_j$, each of which can satisfy the request rate of ${1/(2^{k-2}-2)}$ for file $f_i$ without violating the capacity constraints. Thus, the demand vector ${\tfrac{2^{k-1}+1}{3}\cdot \vek{e}_i+1 \cdot \vek{e}_j}$ can be achieved. Therefore, for ${k=4}$ the demand vector ${3\cdot \vek{e}_i+1 \cdot \vek{e}_j}$ for all ${i,j \in [k]}$ with $i\neq j$ can be satisfied. For achieving the demand vector $\tfrac{4}{3}\cdot \vek{1}$, one can use all the systematic recovery sets $\{\vek{e}_1\}$, $\{\vek{e}_2\}$, $\{\vek{e}_3\}$, $\{\vek{e}_4\}$ with capacity $1$. Moreover, the remaining four columns can be used to build up four recovery sets consisting of a unique recovery set of cardinality $3$ for each file $f_i$, $i \in [4]$, and from each of these recovery sets the rate of $\tfrac{1}{3}$ can be satisfied. This completes the proof.
\end{proof}

\section*{Acknowledgment}
Part of this research is based upon work supported by the National Science Foundation under Grant No.~CIF-1717314.

\bibliographystyle{IEEEtran}
\bibliography{ISIT-2020-Long-Version}

\begin{thebibliography}{10}
\providecommand{\url}[1]{#1}
\csname url@samestyle\endcsname
\providecommand{\newblock}{\relax}
\providecommand{\bibinfo}[2]{#2}
\providecommand{\BIBentrySTDinterwordspacing}{\spaceskip=0pt\relax}
\providecommand{\BIBentryALTinterwordstretchfactor}{4}
\providecommand{\BIBentryALTinterwordspacing}{\spaceskip=\fontdimen2\font plus
\BIBentryALTinterwordstretchfactor\fontdimen3\font minus
  \fontdimen4\font\relax}
\providecommand{\BIBforeignlanguage}[2]{{%
\expandafter\ifx\csname l@#1\endcsname\relax
\typeout{** WARNING: IEEEtran.bst: No hyphenation pattern has been}%
\typeout{** loaded for the language `#1'. Using the pattern for}%
\typeout{** the default language instead.}%
\else
\language=\csname l@#1\endcsname
\fi
#2}}
\providecommand{\BIBdecl}{\relax}
\BIBdecl

\bibitem{noori2016storage}
M.~Noori, E.~Soljanin, and M.~Ardakani, ``On storage allocation for maximum
  service rate in distributed storage systems,'' in \emph{2016 IEEE
  International Symposium on Information Theory (ISIT)}.\hskip 1em plus 0.5em
  minus 0.4em\relax IEEE, 2016, pp. 240--244.

\bibitem{aktacs2017service}
M.~Akta{\c{s}}, S.~E. Anderson, A.~Johnston, G.~Joshi, S.~Kadhe, G.~L.
  Matthews, C.~Mayer, and E.~Soljanin, ``On the service capacity region of
  accessing erasure coded content,'' in \emph{2017 55th Annual Allerton
  Conference on Communication, Control, and Computing (Allerton)}.\hskip 1em
  plus 0.5em minus 0.4em\relax IEEE, 2017, pp. 17--24.

\bibitem{anderson2018service}
S.~E. Anderson, A.~Johnston, G.~Joshi, G.~L. Matthews, C.~Mayer, and
  E.~Soljanin, ``Service rate region of content access from erasure coded
  storage,'' in \emph{2018 IEEE Information Theory Workshop (ITW)}.\hskip 1em
  plus 0.5em minus 0.4em\relax IEEE, 2018, pp. 1--5.

\bibitem{peng2018distributed}
P.~Peng and E.~Soljanin, ``On distributed storage allocations of large files
  for maximum service rate,'' in \emph{2018 56th Annual Allerton Conference on
  Communication, Control, and Computing (Allerton)}.\hskip 1em plus 0.5em minus
  0.4em\relax IEEE, 2018, pp. 784--791.

\bibitem{dimakis2010network}
A.~G. Dimakis, P.~B. Godfrey, Y.~Wu, M.~J. Wainwright, and K.~Ramchandran,
  ``Network coding for distributed storage systems,'' \emph{IEEE Transactions
  on Information Theory}, vol.~56, no.~9, pp. 4539--4551, 2010.

\bibitem{dimakis2011survey}
A.~G. Dimakis, K.~Ramchandran, Y.~Wu, and C.~Suh, ``A survey on network codes
  for distributed storage,'' \emph{Proceedings of the IEEE}, vol.~99, no.~3,
  pp. 476--489, 2011.

\bibitem{huang2013pyramid}
C.~Huang, M.~Chen, and J.~Li, ``Pyramid codes: Flexible schemes to trade space
  for access efficiency in reliable data storage systems,'' \emph{ACM
  Transactions on Storage (TOS)}, vol.~9, no.~1, p.~3, 2013.

\bibitem{gopalan2012locality}
P.~Gopalan, C.~Huang, H.~Simitci, and S.~Yekhanin, ``On the locality of
  codeword symbols,'' \emph{IEEE Transactions on Information Theory}, vol.~58,
  no.~11, pp. 6925--6934, 2012.

\bibitem{sardari2010memory}
M.~Sardari, R.~Restrepo, F.~Fekri, and E.~Soljanin, ``Memory allocation in
  distributed storage networks,'' in \emph{2010 IEEE International Symposium on
  Information Theory}.\hskip 1em plus 0.5em minus 0.4em\relax IEEE, 2010, pp.
  1958--1962.

\bibitem{shanmugam2013femtocaching}
K.~Shanmugam, N.~Golrezaei, A.~G. Dimakis, A.~F. Molisch, and G.~Caire,
  ``Femtocaching: Wireless content delivery through distributed caching
  helpers,'' \emph{IEEE Transactions on Information Theory}, vol.~59, no.~12,
  pp. 8402--8413, 2013.

\bibitem{maddah2016coding}
M.~A. Maddah-Ali and U.~Niesen, ``Coding for caching: fundamental limits and
  practical challenges,'' \emph{IEEE Communications Magazine}, vol.~54, no.~8,
  pp. 23--29, 2016.

\bibitem{hamidouche2014many}
K.~Hamidouche, W.~Saad, and M.~Debbah, ``Many-to-many matching games for
  proactive social-caching in wireless small cell networks,'' in \emph{2014
  12th International Symposium on Modeling and Optimization in Mobile, Ad Hoc,
  and Wireless Networks (WiOpt)}.\hskip 1em plus 0.5em minus 0.4em\relax IEEE,
  2014, pp. 569--574.

\bibitem{aktas2019load}
M.~F. Aktas, A.~Behrouzi-Far, E.~Soljanin, and P.~Whiting, ``Load balancing
  performance in distributed storage with regular balanced redundancy,''
  \emph{arXiv preprint arXiv:1910.05791}, 2019.

\bibitem{allocation:sardariRFS10}
M.~Sardari, R.~Restrepo, F.~Fekri, and E.~Soljanin, ``Memory allocation in
  distributed storage networks,'' in \emph{2010 IEEE International Symposium on
  Information Theory}, June 2010, pp. 1958--1962.

\bibitem{allocation:LeongDH12}
D.~Leong, A.~G. Dimakis, and T.~Ho, ``Distributed storage allocations,''
  \emph{{IEEE} Trans. Information Theory}, vol.~58, no.~7, pp. 4733--4752,
  2012.

\bibitem{latency:JoshiSW15}
G.~Joshi, E.~Soljanin, and G.~W. Wornell, ``Efficient replication of queued
  tasks for latency reduction in cloud systems,'' in \emph{53rd Annual Allerton
  Conference on Communication, Control, and Computing}, 2015, pp. 107--114.

\bibitem{latency:JoshiSW17}
------, ``Efficient redundancy techniques for latency reduction in cloud
  systems,'' \emph{{TOMPECS}}, vol.~2, no.~2, pp. 12:1--12:30, 2017.

\bibitem{Tompecs:AktasKS}
M.~F. Akta{\c{s}}, S.~Kadhe, E.~Soljanin, and A.~Sprintson, ``Analyzing the
  download time of availability codes,'' \emph{ACM Transactions on Modeling and
  Performance Evaluation of Computing Systems}, December 2019.

\bibitem{ISIT:AktasS18}
M.~F. Akta{\c{s}} and E.~Soljanin, ``Heuristics for analyzing download time in
  {MDS} coded storage systems,'' in \emph{2018 IEEE International Symposium on
  Information Theory (ISIT)}.\hskip 1em plus 0.5em minus 0.4em\relax IEEE,
  2018.

\bibitem{aktas2017simplex}
M.~F. Aktas, E.~Najm, and E.~Soljanin, ``Simplex queues for hot-data
  download,'' in \emph{ACM SIGMETRICS Performance Evaluation Review}, vol.~45,
  no.~1.\hskip 1em plus 0.5em minus 0.4em\relax ACM, 2017, pp. 35--36.

\bibitem{rockafellar1970convex}
R.~T. Rockafellar, \emph{Convex analysis}.\hskip 1em plus 0.5em minus
  0.4em\relax Princeton University Press, 1970, vol.~28.

\bibitem{matousek2007understanding}
J.~Matousek and B.~G{\"a}rtner, \emph{Understanding and using linear
  programming}.\hskip 1em plus 0.5em minus 0.4em\relax Springer Science \&
  Business Media, 2007.

\bibitem{tsfasman1995geometric}
M.~A. Tsfasman and S.~G. Vladut, ``Geometric approach to higher weights,''
  \emph{IEEE Transactions on Information Theory}, vol.~41, no.~6, pp.
  1564--1588, 1995.

\bibitem{dodunekov1998codes}
S.~Dodunekov and J.~Simonis, ``Codes and projective multisets,'' \emph{The
  Electronic Journal of Combinatorics}, vol.~5, no.~1, p.~37, 1998.

\bibitem{beutelspacher1998projective}
A.~Beutelspacher, B.~Albrecht, and U.~Rosenbaum, \emph{Projective geometry:
  from foundations to applications}.\hskip 1em plus 0.5em minus 0.4em\relax
  Cambridge University Press, 1998.

\bibitem{assmus1994designs}
E.~F. Assmus and J.~D. Key, \emph{Designs and their Codes}.\hskip 1em plus
  0.5em minus 0.4em\relax Cambridge University Press, 1994, no. 103.

\bibitem{arikan2009channel}
E.~Arikan, ``Channel polarization: A method for constructing capacity-achieving
  codes for symmetric binary-input memoryless channels,'' \emph{IEEE
  Transactions on Information Theory}, vol.~55, no.~7, pp. 3051--3073, 2009.

\bibitem{muller1954application}
D.~E. Muller, ``Application of boolean algebra to switching circuit design and
  to error detection,'' \emph{Transactions of the IRE professional group on
  electronic computers}, no.~3, pp. 6--12, 1954.

\bibitem{reed1953class}
I.~S. Reed, ``A class of multiple-error-correcting codes and the decoding
  scheme,'' Massachusetts Inst. of Tech. Lexington Lincoln Lab., Tech. Rep.,
  1953.

\bibitem{jones2004equality}
C.~Jones, E.~C. Kerrigan, and J.~Maciejowski, ``Equality set projection: A new
  algorithm for the projection of polytopes in halfspace representation,''
  Cambridge University Engineering Dept, Tech. Rep., 2004.

\end{thebibliography}

\appendix[Proofs of Lemmas and Corollaries]
\begin{proof}[Proof of Lemma~\ref{lem:convexity}]
It can be easily observed that for every service rate vector $\boldsymbol{\mu}$, setting $\lambda_{i,j}=0$, where $i\in [k]$ and ${j\in [t_i]}$, satisfies the set of constraints in~\eqref{eq:1} for the all-zero demand vector of dimension $k$ denoted by ${{\boldsymbol{0}}=(0,\dots,0) \in \mathbb{R}^k}$. Thus, ${\boldsymbol{0}}$ always belongs to the service rate region $\mathcal{S}(\mathbf{G},\boldsymbol{\mu})$. It proves that the service rate region $\mathcal{S}(\mathbf{G},\boldsymbol{\mu})$ is a non-empty subset of $\mathbb{R}^k_{\geq 0}$. Based on the definition of the convex set, we need to show that for all $\boldsymbol{\lambda}$ and $\boldmath{\tilde{\lambda}}$ in $\mathcal{S}(\mathbf{G},\boldsymbol{\mu})$ and for all ${0 \leq \pi \leq 1}$, all vectors $\pi{\boldsymbol{\lambda}}+(1-\pi){\boldmath{\tilde{\lambda}}}$ are in $\mathcal{S}(\mathbf{G},\boldsymbol{\mu})$. Since ${{\boldsymbol{\lambda}} \in \mathcal{S}(\mathbf{G},\boldsymbol{\mu})}$, there exist $\lambda_{i,j}$'s, where $i\in [k]$ and ${j\in [t_i]}$, that satisfy the set of constraints in~\eqref{eq:1} for the demand vector $\boldsymbol{\lambda}$ and the service rate vector $\boldsymbol{\mu}$. Also, since ${{\boldmath{\tilde{\lambda}}} \in \mathcal{S}(\mathbf{G},\boldsymbol{\mu})}$, there exist $\tilde{\lambda}_{i,j}$'s, where $i\in [k]$ and ${j\in [t_i]}$, that satisfy the set of constraints in~\eqref{eq:1} for the demand vector ${\boldmath{\tilde{\lambda}}}$ and the service rate vector $\boldsymbol{\mu}$. One can easily confirm that $({\pi{\lambda_{i,j}}+(1-\pi){\tilde{\lambda}_{i,j}}})$'s, where $i\in [k]$ and ${j\in [t_i]}$, also satisfy the set of constraints in~\eqref{eq:1} for the demand vector ${\pi{\boldsymbol{\lambda}}+(1-\pi){\boldmath{\tilde{\lambda}}}}$ for all $0 \leq \pi \leq 1$, and the service rate vector $\boldsymbol{\mu}$. Thus, $\pi{\boldsymbol{\lambda}}+(1-\pi){\boldmath{\tilde{\lambda}}}$ belongs to $\mathcal{S}(\mathbf{G},\boldsymbol{\mu})$ for all $0 \leq \pi \leq 1$. This completes the proof of convexity of the service rate region $\mathcal{S}(\mathbf{G},\boldsymbol{\mu})$. Summing up the set of constraints in~\eqref{condition_capacity} leads us to:
\[\sum_{l=1}^{n}\sum_{i=1}^{k}~\sum_{\substack{{j\in [t_i]} \\ {l \in R_{i,j}}}} \lambda_{i,j} \leq \sum_{l=1}^{n} \mu_l\]
Changing the order of the sums and utilizing the fact that ${\sum_{l=1}^{n}~\sum_{\substack{{j\in [t_i]} \\ {l \in R_{i,j}}}} \lambda_{i,j}=\sum_{j=1}^{t_i} \lambda_{i,j}}$, we obtain 
\[\sum_{i=1}^{k}\sum_{j=1}^{t_i} \lambda_{i,j} \leq \sum_{l=1}^{n} \mu_l.\]
Using~\eqref{condition_demand}, we rewrite the last inequality to
\begin{equation}\label{eq:closed}
    \sum_{i=1}^{k} \lambda_{i} \leq \sum_{l=1}^{n} \mu_l
\end{equation}
The equation~\eqref{eq:closed} indicates that the elements of every vector ${\boldsymbol{\lambda}}\in \mathcal{S}(\mathbf{G},\boldsymbol{\mu})$ are bounded. It also shows that all demand vectors ${\boldsymbol{\lambda}}=(\lambda_{1},\dots,\lambda_{k})$ with $\sum_{i=1}^{k} \lambda_{i} > \sum_{l=1}^{n} \mu_l$ are not in $\mathcal{S}(\mathbf{G},\boldsymbol{\mu})$. Hence, $\mathcal{S}(\mathbf{G},\boldsymbol{\mu})$ is closed and bounded.
\end{proof}

\begin{proof}[Proof of Corollary~\ref{lem:polytope}]
Based on Lemma~\ref{lem:convexity}, the service rate region $\mathcal{S}(\mathbf{G},\boldsymbol{\mu})$ is a convex and bounded subset of the $\mathbb{R}^k_{\geq 0}$, which indicates that $\mathcal{S}(\mathbf{G},\boldsymbol{\mu})$ is a polytope. Thus, according to~\cite[Theorem 4]{jones2004equality}, it can be described as the two mentioned forms, i.e., the intersection of a finite number of half spaces or the convex hull of a finite set of vectors (the vertices of the polytope). 
\end{proof}

\begin{proof}[Proof of Lemma~\ref{lemma_hyperplane_constraint}]
Since $s\cdot \vek{e}_i\in\mathcal{S}(\mathbf{G},\boldsymbol{\mu})$, it means that the request rate of $s$ for file $f_i$ is satisfied by the coded storage system. Whatever the used recovery sets for file $f_i$ are, some points outside of $\mathcal{H}$ have to be used since the points in $\mathcal{H}$ are not able to generate $\vek{e}_i$. Thus, replacing each recovery set in $\mathcal{R}_i$ by an arbitrary contained point outside of hyperplane~$\mathcal{H}$, completes the proof. 
\end{proof}

\begin{proof}[Proof of Corollary~\ref{cor_min_distance}]
Since for all ${i \in [k]}$, ${s\cdot \vek{e}_i\in\mathcal{S}(\mathbf{G},\boldsymbol{\mu})}$ holds, this means that for all files $f_i$, $i \in [k]$, the request rate of $s$ can be satisfied by the coded storage system. Thus, if we consider any hyperplane $\mathcal{H}$ in ${\PG(k-1,q)}$, it does not contain at least one of the $\vek{e}_i$'s for $i \in [k]$. 
In the special case of unit service rate of all servers, based on Lemma~\ref{lemma_hyperplane_constraint} results in 
\begin{equation}
  \nonumber
  s\le \# \left(\mathcal{G}\backslash \mathcal{H}\right):=\#\mathcal{G}-\#\left(\mathcal{G}\cap \mathcal{H}\right)=n-\#\left(\mathcal{G}\cap \mathcal{H}\right).
\end{equation}
Since for every hyperplane $\mathcal{H}$ in ${\PG(k-1,q)}$, $s\le n-\#\left(\mathcal{G}\cap \mathcal{H}\right)$ holds, according to the Proposition~\ref{Prop:mindis} and based on the fact that the minimum distance $d$ is integer, we have ${\lceil s\rceil \leq d}$.
\end{proof}

\begin{proof}[Proof of Corollary~\ref{cor_hyperplane_constraint}]
Since ${\sum_{i\in\mathcal{I}} s_i\cdot \vek{e}_i\in\mathcal{S}(\mathbf{G},\boldsymbol{\mu})}$, based on Lemma~\ref{lem:convexity}, ${s_i\cdot \vek{e}_i\in\mathcal{S}(\mathbf{G},\boldsymbol{\mu})}$ holds for all ${i \in \mathcal{I}}$. On the other hand, the hyperplane $\mathcal{H}$ of ${\PG(k-1,q)}$ does not contain any ${\vek{e}_i}$ for all ${i \in \mathcal{I}}$. Thus, by applying Lemma~\ref{lemma_hyperplane_constraint} for each ${i \in \mathcal{I}}$, we get ${s_i \leq \sum_{p \in {\PG(k-1,q) \setminus \mathcal{H}}} \mu(p)}$. Summing up all these inequalities gives
\[
  s=\sum_{i \in \mathcal{I}} s_i \leq \sum_{p \in {\PG(k-1,q) \setminus \mathcal{H}}} \mu(p).
\]
\end{proof}

\end{document}